\documentclass{amsart}
\pdfoutput=1
\usepackage{amsmath}
\usepackage{amsthm}
\usepackage{amssymb}
\usepackage{esint}
\usepackage{mathtools}
\usepackage{hyperref}
\usepackage[english]{babel}
\usepackage[T1]{fontenc}

\usepackage{pgf,tikz}
\usetikzlibrary{arrows}

\theoremstyle{plain}
\newtheorem{theorem}{Theorem}
\newtheorem{lemma}[theorem]{Lemma}
\newtheorem{corollary}[theorem]{Corollary}
\newtheorem{proposition}[theorem]{Proposition}
\theoremstyle{definition}
\newtheorem{definition}[theorem]{Definition}

\theoremstyle{remark}

\newcommand{\R}{\mathbb R}
\newcommand{\C}{\mathbb C}
\newcommand{\N}{\mathbb N}
\newcommand{\Z}{\mathbb Z}
\newcommand{\Q}{\mathbb Q}
\newcommand{\h}{\mathcal H}
\newcommand{\A}{\mathcal A}

\newcommand{\brt}{{\mathcal G}}
\newcommand{\nbrt}{\tilde{\mathcal G}}

\newcommand{\der}{\mathrm{d}}
\newcommand{\eps}{\varepsilon}
\newcommand{\abs}[1]{\left| #1 \right|}
\newcommand{\aabs}[1]{\left\| #1 \right\|}
\newcommand{\sisus}{\operatorname{int}}

\newcommand{\co}{\operatorname{co}}
\newcommand{\Tr}{\operatorname{tr}}
\newcommand{\diam}{\operatorname{diam}}
\newcommand{\ang}{{\mathrm{ang}}}
\newcommand{\rad}{{\mathrm{rad}}}

\newcommand{\rot}{\mathfrak R}
\newcommand{\radon}{\mathcal R}

\newcommand{\order}{o}
\newcommand{\CDini}{C^{\mathrm{DL}}(\bar{D})}
\newcommand{\DL}{Dini-Lipschitz}
\renewcommand{\phi}{\varphi}
\renewcommand{\theta}{\vartheta}
\begin{document}
\date{\today}
\title{Broken ray tomography in the disk}
\author{Joonas Ilmavirta}
\address{Department of Mathematics and Statistics, University of Jyv\"askyl\"a, P.O.Box 35 (MaD) FI-40014 University of Jyv\"askyl\"a, Finland}
\email{joonas.ilmavirta@jyu.fi}

\begin{abstract}
Given a bounded $C^1$ domain $\Omega\subset\R^n$ and a nonempty subset $E$ of its boundary (set of tomography), we consider broken rays which start and end at points of $E$. We ask: If the integrals of a function over all such broken rays are known, can the function be reconstructed? We give positive answers when $\Omega$ is a ball and the unknown function is required to be uniformly quasianalytic in the angular variable and the set of tomography is open. We also analyze the situation when the set of tomography is a singleton.
\end{abstract}


%\submitto{\IP}

\maketitle

\section{Introduction}

Consider a ray of light travelling to the right on the real axis. If the intensity of the light at a point $x$ is $I(x)$ and the material on the real axis has a non-constant attenuation coefficient $f\geq0$, the differential equation
\begin{equation}
%\label{eq:}
I'(x)=-f(x)I(x)
\end{equation}
is satisfied. Given the initial intensity $I(0)$ and the intensity at a distance $x$, we may easily find (under suitable regularity assumptions) that
\begin{equation}
%\label{eq:}
\log\left(\frac{I(0)}{I(x)}\right)=\int_0^x f(y)\der y.
\end{equation}
That is, by measuring the ratio of initial and final intensity of a beam of light we in effect measure the integral of the attenuation coefficient over its trajectory. This is true even when the trajectory is more complicated than a segment of the real line.

The fundamental problem of (scalar) ray tomography is to deduce knowledge of the attenuation coefficient $f$ from its integrals over a suitable set of trajectories. Different materials (and different wavelengths) have different attenuation coefficients, so information about $f$ can be translated into understanding of the structure of the object being imaged.

Consider a bounded $C^1$ domain $\Omega$ in a Euclidean space and a continuous attenuation coefficient $f:\bar{\Omega}\to\R$.
A piecewise linear path in $\bar{\Omega}$ %starting and ending at the boundary and 
with reflections at the boundary as dictated by geometric optics (the angle of incidence equals the angle of reflection) is called a \emph{broken ray}.
We have a device $E$, which is a nonempty subset of the boundary $\partial\Omega$, from which we may emit light and measure the intensity of incoming light. The rest of the boundary acts as a reflector and we consider those broken rays which start and end at $E$, which we call the \emph{set of tomography}. The corresponding set of integrals of $f$ over such broken rays is called the \emph{broken ray transform of $f$} with the set of tomography $E$. We will implicitly assume all broken rays to start and end at the set of tomography.

We ask the following question: If the integral of $f$ is known over all broken rays with both endpoints in $E$, can the function $f$ be fully reconstructed? How does the result depend on the domain $\Omega$, the set of tomography $E$, and regularity assumptions on the unknown function $f$?

In the case $E=\partial\Omega$ broken rays are simply intersections of lines with the domain $\Omega$. By extending the unknown function $f$ by zero to $\R^n\setminus\bar{\Omega}$, the problem reduces to the classical X-ray transform. This transform is well understood: The answer to the uniqueness question is affirmative, and also regularity and feasible reconstruction algorithms are known. For these results we refer to the textbooks by Helgason~\cite{book-helgason} and Natterer~\cite{book-natterer} and the original works of Radon~\cite{radon} and Cormack~\cite{cormack}.

Reflections from a boundary of a domain naturally arise in connection with billiards. For details we refer to the textbooks by Tabachnikov~\cite{book-billiards} and Chernov and Markarian~\cite{book-cm}. Reflections in connection to the X-ray transform as considered here, however, appear to have been studied less.

Florescu, Markel, and Schotland~\cite{schotland} have studied the broken ray Radon transform in connection with the problem of recovering absorption and scattering coefficients in single scattering tomography.
Eskin~\cite{eskin} has studied the reconstruction of electromagnetic potential in the presence of convex reflecting obstacles.
Lozev~\cite{lozev1} has also studied the numerical aspect of scalar tomography with similar obstacles.
The question asked above has also turned out to be related to the Calder\'on problem with partial data~\cite{KS:calderon}.

It should be noted that in optical tomography the signal is very weak after multiple reflections, and reconstruction schemes should be based on broken rays with only few reflections. One such method is imaging by Compton cameras~\cite{BZG:compton1,BZG:compton2,ADHKK}. The broken ray transform with one reflection is also known as the V-line Radon transform, which has recently been studied and applied to imaging problems~\cite{MNTZ:radon,TN:radon,A:radon}.

The main results presented here require knowledge of the broken ray transform with arbitrarily many reflections, and can therefore not be expected to be practical in optical tomography. Our main motivation is the Calder\'on problem with partial data; on suitable manifolds, reconstructibility of a conductivity from partial Dirichlet to Neumann data can be reduced to the injectivity of the broken ray transform~\cite{KS:calderon}. 
%In this application increasing the number of reflections is tolerable.
The broken ray transform arises similarly in partial data problems for the electromagnetic Schr\"odinger equation~\cite{eskin}.% and the wave equation~\cite{???}.

In this article we consider the special case where $\Omega$ is the unit disk $D\subset\R^2$, and give the following answers:

\begin{theorem}
\label{thm:one}
If $f:\bar{D}\to\R$ is continuous %and satisfies the Dini condition
and the set of tomography is a singleton, then integrals over broken rays uniquely determine the integral of $f$ over any circle centered at the origin.
\end{theorem}

\begin{corollary}
\label{cor:one}
Let $B^n$ denote the unit ball in $\R^n$. If $n\geq 2$, $f:\bar{B}^n\to\R$ is continuous, %and satisfies the Dini condition,
and the set of tomography is a singleton, then integrals over broken rays uniquely determine 
%the value of $f$ at the origin and 
the integral of $f$ over any circle centered at the origin with the singleton in the circle's plane.
\end{corollary}

Full reconstruction is guaranteed for functions that are uniformly quasianalytic in the angular variable in the following theorem and its corollary. Uniformly quasianalytic functions are continuous in the radial variable and quasianalytic\footnote{For example, real-analytic functions are quasianalytic.} in the angular variable with a uniformity condition given in definition~\ref{def:qa}. There is, in particular, great freedom in the radial dependence of the unknown function. Assuming quasianalyticity in both variables makes the problem much easier, as demonstrated in corollary~\ref{cor:anal}.

\begin{theorem}
\label{thm:open}
If $f:\bar{D}\to\R$ is uniformly quasianalytic in the angular variable in the sense of definition~\ref{def:qa} and its angular derivatives of all orders satisfy the \DL\ condition and the set of tomography is open, then integrals over broken rays uniquely determine $f$ everywhere in the disk.
\end{theorem}

\begin{corollary}
\label{cor:open}
Theorem~\ref{thm:open} also holds when $D$ is replaced by the unit ball in any Euclidean space of dimension two or higher.
\end{corollary}

We need not assume any regularity of $f$ at the origin to guarantee reconstruction in the punctured disk $\bar{D}\setminus\{0\}$ in theorems~\ref{thm:one} and~\ref{thm:open}. If we assume continuity at the origin, theorem~\ref{thm:one} implies that $f(0)$ can be reconstructed.

One class of functions satisfying the regularity assumptions of theorem~\ref{thm:open} are functions of the form (in polar coordinates)
\begin{equation}
%\label{eq:}
f(r,\theta)=a_0(r)+\sum_{k=1}^K (a_k(r)\cos(k\theta)+b_k(r)\sin(k\theta)),
\end{equation}
where $K\in\N$ and the functions $a_k,b_k$ are H\"older continuous.%satisfy the \DL\ condition.

We also show in lemma~\ref{lma:rot} that if theorem~\ref{thm:open} would hold for all functions which are $C^\infty$ in the angular variable, then it would also hold for all continuous functions. Whether theorem~\ref{thm:open} indeed holds for all continuous function remains an interesting open question.

We begin by making some general remarks about the nature of the problem and presenting the formalism for the broken ray transform in the disk (section~\ref{sec:prel}). In section~\ref{sec:one-pf} we prove theorem~\ref{thm:one} and corollary~\ref{cor:one}; singleton tomography is discussed in more detail in section~\ref{sec:one}. After developing the necessary tools in sections~\ref{sec:int} and~\ref{sec:tools}, we present the 
proof of theorem~\ref{thm:open} and corollary~\ref{cor:open} in section~\ref{sec:open}.
%proofs of theorems~\ref{thm:one} and~\ref{thm:open} and their corollaries in sections~\ref{sec:one} and~\ref{sec:open}.

\section{Preliminaries}
\label{sec:prel}

\subsection{Some initial remarks}

Let us first consider the general problem with a bounded $C^1$ Euclidean domain $\Omega$.
The problem is linear in the function $f$, so we may reformulate the uniqueness question of the introduction as follows: If the integral of $f$ over all broken rays vanishes, is $f$ identically zero?

Linearity also causes the problem for a complex-valued $f$ to reduce to the real-valued case: the problems of finding the real and imaginary parts of $f$ decouple.

Decreasing the size of $E$ makes the problem more difficult by decreasing the number of broken rays available for reconstruction. If it is convenient to consider some trajectories that have reflections also on $E$, we may do so. Summing the integrals over each part of the trajectory starting and ending at $E$, we may construct the integral over a trajectory with reflections on $E$. Thus the only requirement for the broken rays of interest is that their endpoints lie in the set of tomography $E$.

The assumption of $C^1$ regularity of $\partial\Omega$ is necessary to make sense of the reflections everywhere. A natural assumption on the unknown function $f:\bar{\Omega}\to\R$ is continuity. Corollary~\ref{cor:anal} demonstrates that too strong regularity assumptions make the problem very easy.

%; it could be loosened, but is already physically reasonable. Corollary~\ref{cor:anal} demonstrates that too heavy assumptions make the problem meaningless: global behaviour of $f$ should not be inferable from local behaviour.\footnote{If an X-ray image of a toe shows that the toe is healthy, must the heart be healthy, too?}

Convexity of $\Omega$ is not required in the following proposition, but we keep the assumption to make the proof simpler.

\begin{proposition}
\label{prop:spt}
%Let $\Omega\subset\R^n$, $n\geq2$, be a bounded convex domain. If $f:\bar{\Omega}\to\R$ is continuous and the integral of $f$ over all broken rays vanishes, then $f$ vanishes in the interior of the convex hull of each connected component of the set of tomography~$E$.
Let $\Omega\subset\R^n$, $n\geq2$, be a bounded convex domain. If $f:\bar{\Omega}\to\R$ is continuous and the integral of $f$ over all broken rays vanishes, then $f$ vanishes
in the exterior of the convex hull of the complement of the union of the convex hulls of the components of the set of tomography~$E$.
This set is an intersection of $\Omega$ with a neighborhood of $E$ in $\bar{\Omega}$ if $\Omega$ is strictly convex and $E\subset\partial\Omega$ is open.
\end{proposition}

Let $E_i$, $i=1,\dots,m$, be the connected components of $E$. The first claim is that $f$ vanishes in the exterior of $U=\co(\Omega\setminus\bigcup_{i=1}^m\co E_i)$, where $\co A$ is the convex hull of a set~$A$. In dimension two $\Omega\setminus\bigcup_{i=1}^m\co E_i$ is convex, so $f$ vanishes in $\sisus\bigcup_{i=1}^m\co E_i$.

\begin{proof}[Proof of proposition~\ref{prop:spt}]
Let $\eps>0$ and denote by $f_\eps$ the function $f$ convolved with the standard mollifier with support in $B(0,\eps)$.
The integral of $f$ vanishes over any line that does not meet $U$, whence the integral of $f_\eps$ vanishes over any line that does not meet the convex set $B(U,\eps)=\{x\in\R^n:d(x,U<\eps)\}$. Since $f_\eps\in C_0^\infty(\R^n)$, the function $f_\eps$ vanishes outside $B(U,\eps)$ by Helgason's support theorem~\cite[Theorem~2.6]{book-helgason}.

Take any point $y\in\sisus(\Omega\setminus U)$. For all $\eps<d(y,U)$ we know that $f_\eps(y)=0$. Since $f$ is continuous at $y$, $f_\eps(y)$ converges to $f(y)$ as $\eps\to0$. Thus $f$ vanishes in the exterior of $U$ as claimed.

For the second claim, fix any point $a\subset E$. To prove the claim, we find a neighborhood $V$ of $a$ in the topology of $\bar{\Omega}$ such that $V\cap\Omega\subset\sisus(\Omega\setminus U)$.

Let $E_i$ be the component of $E$ containing $a$. By convexity there is a non-degenerate linear form $L:\R^n\to\R$ and a constant $c\in\R$ such that $\Omega\subset\{x\in\R^n:L(x)<c\}$ and $L(a)=c$. By strict convexity of $\Omega$ and openness of $E$ there is $\eps>0$ such that for $V=\{x\in\bar{\Omega}:L(x)>c-\eps\}$ we have $V\cap\partial\Omega\subset E$. Since $\Omega\setminus V$ is convex and $V\subset\co E_i$, we have $U\cap V=\emptyset$. But $V$ is also open, so $V\cap\Omega\subset\sisus(\Omega\setminus U)$.
\end{proof}

\begin{corollary}
\label{cor:anal}
%Suppose the convex hull of a connected component of $E$ contains an interior point of $\Omega$.
Suppose $\Omega\subset\R^n$, $n\geq2$, is a bounded strictly convex domain and the set of tomography $E\subset\partial\Omega$ is open. If $f$ is quasianalytic in $\Omega$ and continuous in $\bar{\Omega}$ and the integral of $f$ over all broken rays vanishes, then $f$ is identically zero.
\end{corollary}
\begin{proof}
%We may again assume that $E$ is connected. Since $E$ is open and $\Omega$ is strictly convex, $\sisus\co E\neq\emptyset$. By proposition~\ref{prop:spt} $f$ vanishes in this nonempty open subset of $\Omega$, from which the claim follows by quasianalyticity of $f$.
By proposition~\ref{prop:spt} $f$ vanishes in a nonempty open subset of $\Omega$, from which the claim follows by quasianalyticity of $f$.
\end{proof}

Note that in the preceeding proposition and corollary only the non-reflecting trajectories were used, although their union is not even dense in $\Omega$. %The proposition demonstrates, however, that under mild assumptions (such as continuity) Helgason's theorem guarantees that $f$ vanishes in $\co E$ if $E$ is connected (if $E$ is not, $f$ vanishes in $\bar{\Omega}\setminus\co(\partial\Omega\setminus E)$), but does \emph{not} allow us to `see' all of $\Omega$.

\subsection{Formalism for the broken ray transform in the disk}

We use polar coordinates $(r,\theta)$ in the unit disk $D\subset\R^2$ and identify points on the boundary $S_1=\partial D$ with the corresponding angle. The notation $E=\{0\}$ means that $E$ is one point at angle zero.

A broken ray $\gamma$ is characterized by the number of line segments $n_\gamma$ (there are $n_\gamma-1$ reflections), the initial point $\iota_\gamma\in E$, the final point $\kappa_\gamma\in E$, the angle $\alpha_\gamma$ at which a line segment appears when seen from the origin, and an integer winding number $m_\gamma$ such that the trajectory condition
\begin{equation}
\label{eq:trajectory}
n_\gamma \alpha_\gamma = \kappa_\gamma-\iota_\gamma+2\pi m_\gamma
\end{equation}
is satisfied.
These five parameters $(n_\gamma,m_\gamma,\alpha_\gamma,\iota_\gamma,\kappa_\gamma)$ uniquely (and redundantly) define the broken ray $\gamma$.
We require that $n_\gamma>0$ but the other parameters may be negative.
We denote by $\Gamma_E$ the set of all broken rays from $E$ to $E$, allowing intermediate reflections on $E$. We denote by $z_\gamma=\cos(\alpha_\gamma/2)$ and $d_\gamma=2\abs{\sin(\alpha_\gamma/2)}$ the distance from the origin to the trajectory and the length of each individual line segment. The parameters describing the broken ray are illustrated in figure~\ref{fig:br-param}.

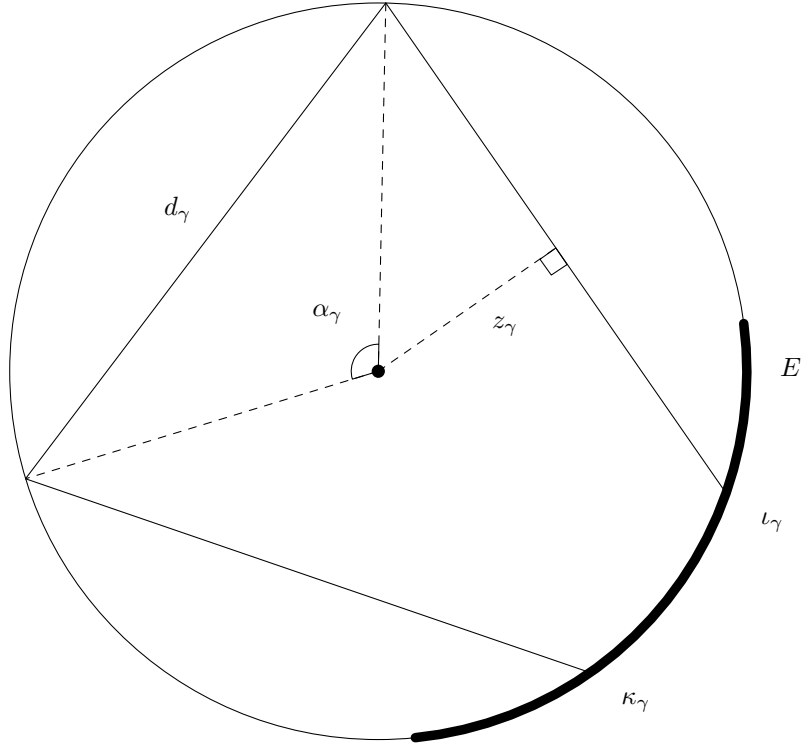
\begin{figure}%
%\begin{center}
\begin{tikzpicture}[line cap=round,line join=round,>=triangle 45,x=1.0cm,y=1.0cm]
\clip(-10.55,-0.11) rectangle (0.48,10.1);
\draw [shift={(-5.39,5.01)}] (0,0) -- (88.78:0.36) arc (88.78:196.93:0.36) -- cycle;
\draw(-3.24,6.5) -- (-3.09,6.29) -- (-2.88,6.43) -- (-3.03,6.65) -- cycle; 
\draw (-6.39,6.02) node[anchor=north west] {$ \alpha_\gamma $};
\draw(-5.39,5.01) circle (4.9cm);
\draw [dash pattern=on 3pt off 3pt] (-5.39,5.01)-- (-3.03,6.65);
\draw (-0.76,3.38)-- (-5.29,9.91);
\draw (-5.29,9.91)-- (-10.08,3.58);
\draw (-10.08,3.58)-- (-2.57,1);
\draw [dash pattern=on 3pt off 3pt] (-5.29,9.91)-- (-5.39,5.01);
\draw [dash pattern=on 3pt off 3pt] (-5.39,5.01)-- (-10.08,3.58);
\draw [shift={(-5.39,5.01)},line width=3.6pt]  plot[domain=-1.47:0.13,variable=\t]({1*4.9*cos(\t r)+0*4.9*sin(\t r)},{0*4.9*cos(\t r)+1*4.9*sin(\t r)});
\draw (-4,5.86) node[anchor=north west] {$ z_\gamma $};
\draw (-0.43,3.25) node[anchor=north west] {$ \iota_\gamma $};
\draw (-2.28,0.86) node[anchor=north west] {$ \kappa_\gamma $};
\draw (-8.36,7.46) node[anchor=north west] {$ d_\gamma $};
\draw (-0.17,5.31) node[anchor=north west] {$ E $};
\begin{scriptsize}
\fill [color=black] (-5.39,5.01) circle (2.5pt);
\end{scriptsize}
\end{tikzpicture}
\caption{An example of a broken ray $\gamma$ in the disk. Here $m_\gamma=1$ and $n_\gamma=3$. The parameters $\alpha_\gamma$, $z_\gamma$, $d_\gamma$, $\iota_\gamma$, and $\kappa_\gamma$ as well as the set of tomography $E$ are illustrated in the figure. The angles $\iota_\gamma$ and $\kappa_\gamma$ are identified with the corresponding points on the boundary.}%
\label{fig:br-param}%
%\end{center}
\end{figure}

We also write
\begin{equation}
%\label{eq:}
\Gamma_E^N = \{\gamma\in\Gamma_E:n_\gamma\geq N\}
\end{equation}
to denote the set of rays which have sufficiently many reflections.

A trajectory $\gamma$ is a (piecewise linear) curve, and we denote its trace by $\Tr(\gamma)$.

Let $C=C(\bar{D};\R)$ be the space of continous real-valued mappings from the closed unit disk, $\CDini\subset C$ the subspace of functions which satisfy the \DL\ condition, and $B(\Gamma_E,\R)$ the space of bounded real-valued mappings from $\Gamma_E$.
The mappings in $B(\Gamma_E,\R)$ need not be linear, continuous, or measurable (the set $\Gamma_E$ is not equipped with any such structure), but merely bounded.
We equip all of these spaces with the supremum norm. The \DL\ condition states that the modulus of continuity $\omega_f$ of a function $f$ satisfy $\omega_f(\delta)=\order(\log^{-1}\delta^{-1})$ for small~$\delta$.
H\"older continuous functions, for example, satisfy the \DL\ condition.

\begin{definition}
For $f\in C$, we define the \emph{broken ray transform of $f$} as the function $\brt f:\Gamma_E\to\R$ such that
\begin{equation}
%\label{eq:}
\brt f(\gamma) = \fint_\gamma f \der\h^1 = \frac{1}{l}\int_0^l f(\gamma(t))\der t,
\end{equation}
when $\gamma:[0,l]\to\bar{D}$ is a broken ray with unit speed.
The map $\brt:C\to B(\Gamma_E,\R)$ is the \emph{broken ray transform}.
%We also define the unnormalized broken ray transform $\nbrt$ such that
We define the unnormalized broken ray transform $\nbrt$ by
\begin{equation}
%\label{eq:}
\nbrt f(\gamma) = \int_\gamma f \der\h^1.
\end{equation}
\end{definition}

%If $f\in C$, then $\brt f$ as defined above is indeed an element of $B(\Gamma_E,\R)$.
Elementary calculations show that $\brt$ is linear and continuous. Using a constant function as an example, one may show that moreover $\aabs{\brt}=1$. The unnormalized version, however, is not continuous for generic $E$: the unnormalized broken ray transform of a constant function is not bounded, as can be seen by increasing the length of the trajectory indefinitely.

\subsection{Proof of theorem 1}
\label{sec:one-pf}

We are now ready to prove theorem~\ref{thm:one}. A Fourier analytic proof for the special case when $f\in\CDini$ is given in section~\ref{sec:one}, where the result is also discussed in more detail.

\begin{proof}[Proof of theorem~\ref{thm:one}]
We denote circles centered at the origin by $S_r=S(0,r)$.
%As in lemma~\ref{lma:Fourier}, we denote the average of $f$ over such circles by $a_0(r)=\fint_{S_r}f\der\h^1$.
We denote the average of $f$ over such circles by $a_0(r)=\fint_{S_r}f\der\h^1$.
%By corollary~\ref{cor:calc} it suffices to consider those $\gamma$ that have $\gcd(n_\gamma,m_\gamma)=1$.
It suffices to consider those $\gamma$ that have $\gcd(n_\gamma,m_\gamma)=1$.
We may assume that $\iota_\gamma=\kappa_\gamma=0$.

%A simple change of variable as in the proof of lemma~\ref{lma:calc} yields
Denoting $\phi_l^\pm(r)=(l-\frac{1}{2})\alpha_\gamma\pm\arccos(z_\gamma/r)$, we have
\begin{equation}
%\label{eq:}
\begin{split}
\nbrt f(\gamma)
&=
%\int_\gamma f\der\h^1
%=
%\sum_{l=1}^{n_\gamma}\int_{-\alpha_\gamma/2}^{\alpha_\gamma/2}f\left(\frac{\cos(\alpha_\gamma/2)}{\cos(\phi)},(l-\frac{1}{2})\alpha_\gamma+\phi\right)\frac{\cos(\alpha_\gamma)}{\cos^2(\phi)}\der\phi
\sum_{l=1}^{n_\gamma}\int_{z_\gamma}^1 (f(r,\phi_l^+(r))+f(r,\phi_l^-(r)))\frac{r}{\sqrt{r^2-z_\gamma^2}}\der r
\\&=
\frac{d_\gamma}{2}
\int_{z_\gamma}^1 p_{z_\gamma}(r) \sum_{y\in S_r\cap\Tr(\gamma)} f(y)\der r,
\end{split}
\end{equation}
where we have defined the function $p_z:(z,1]\to\R$ for any $z\in(0,1)$ such that
\begin{equation}
%\label{eq:}
p_z(r)=\left((1-z^2)(1-z^2/r^2)\right)^{-1/2}.
\end{equation}
Thus
\begin{equation}
%\label{eq:}
\brt f(\gamma)
=
\frac{1}{n_\gamma d_\gamma}\nbrt f(\gamma)
=
\int_{z_\gamma}^1 p_{z_\gamma}(r) \frac{1}{2n_\gamma}\sum_{y\in S_r\cap\Tr(\gamma)} f(y)\der r.
\end{equation}
The function $p_{z_\gamma}$ has a clear geometric interpretation: it is the probability distribution for the distance from the origin on the trajectory $\gamma$. Therefore we expect that $\brt f(\gamma)$ approaches $g(z_\gamma)$,
\begin{equation}
\label{eq:g-def}
g(z) = \int_z^1 p_z(r) a_0(r) \der r,
\end{equation}
when $n_\gamma\to\infty$ in a suitable sense.

Let us fix some $\gamma\in\Gamma_E^3$. The set $S_r\cap\Tr(\gamma)$ has $2n_\gamma$ elements for almost every $r\in(z_\gamma,1)$. For any such $r$ we may write
\begin{equation}
%\label{eq:}
S_r\cap\Tr(\gamma) = \{x_1^{\gamma,r},\dots,x_{2n_\gamma}^{\gamma,r}\}
\end{equation}
and assume that the points $x_i^{\gamma,r}$ are numbered clockwise on $S_r$.
%We may pick any point in $S_r\cap\Tr(\gamma)$ to be $x_1^{\alpha,r}$ and let the rest be numbered clockwise from that point.
%We may assume that the points $x_i^{\alpha,r}$ are numbered clockwise from any point in $S_r$.
We split $S_r$ to $2n_\gamma$ pairwise disjoint arcs $I_i^{\gamma,r}$ such that each arc corresponds to the angle $\pi/n_\gamma$ and $x_i^{\gamma,r} \in I_i^{\gamma,r}$.
This can be done since the points $x_i^{\gamma,r}$ are pairwise uniformly distributed on $S_r$ in the sense that a relabeling $x_i^{\gamma,r} \mapsto x_{i+2}^{\gamma,r}$ corresponds to a rotation by an angle $2\pi/n_\gamma$. The length of each arc $I_i^{\gamma,r}$ is  $\h^1(I_i^{\gamma,r})=\pi r/n_\gamma$, so the diameter is $\diam(I_i^{\gamma,r})<\pi r/n_\gamma$.
The points $x_i^{\gamma,r}$ and arcs $I_i^{\gamma,r}$ are illustrated in figure~\ref{fig:proof-one}.

\begin{figure}%
%\begin{center}
\begin{tikzpicture}[line cap=round,line join=round,>=triangle 45,x=1.0cm,y=1.0cm]
\clip(-1.77,-7.17) rectangle (11.92,5.5);
\draw(4.42,-0.74) circle (5.7cm);
\draw(4.42,-0.74) circle (3.39cm);
\draw (10.12,-0.67)-- (-0.23,2.55);
\draw (-0.23,2.55)-- (6.25,-6.14);
\draw (6.25,-6.14)-- (6.11,4.7);
\draw (6.11,4.7)-- (-0.15,-4.15);
\draw (-0.15,-4.15)-- (10.12,-0.67);
\draw (10.38,-0.92) node[anchor=north west] {$E$};
\draw (6.79,-0.09) node[anchor=north west] {$x_1^{\gamma,r}$};
\draw [shift={(4.42,-0.74)},line width=3.6pt]  plot[domain=0.01:0.64,variable=\t]({1*3.39*cos(\t r)+0*3.39*sin(\t r)},{0*3.39*cos(\t r)+1*3.39*sin(\t r)});
\draw (7.75,0.77) node[anchor=north west] {$I_1^{\gamma,r}$};
\draw (7.83,-1.67) node[anchor=north west] {$x_2^{\gamma,r}$};
\draw (6.45,-3.70) node[anchor=north west] {$x_3^{\gamma,r}$};
\draw (6.38,2.52) node[anchor=north west] {$x_{10}^{\gamma,r}$};
\begin{scriptsize}
\fill [color=black] (4.42,-0.74) circle (2.5pt);
\fill [color=black] (10.12,-0.67) circle (4.5pt);
\fill [color=black] (7.14,1.29) circle (2.0pt);
\fill [color=black] (1.7,-2.77) circle (2.0pt);
\fill [color=black] (3.41,-3.98) circle (2.0pt);
\fill [color=black] (5.43,2.5) circle (2.0pt);
\fill [color=black] (3.33,2.47) circle (2.0pt);
\fill [color=black] (5.51,-3.95) circle (2.0pt);
\fill [color=black] (7.19,-2.7) circle (2.0pt);
\fill [color=black] (1.65,1.22) circle (2.0pt);
\fill [color=black] (7.81,-0.7) circle (2.0pt);
\fill [color=black] (1.03,-0.78) circle (2.0pt);
\draw [color=black] (7.71,0.08)-- ++(-2.5pt,-2.5pt) -- ++(5.0pt,5.0pt) ++(-5.0pt,0) -- ++(5.0pt,-5.0pt);
\draw [color=black] (2.18,1.8)-- ++(-2.5pt,-2.5pt) -- ++(5.0pt,5.0pt) ++(-5.0pt,0) -- ++(5.0pt,-5.0pt);
\draw [color=black] (6.22,-3.62)-- ++(-2.5pt,-2.5pt) -- ++(5.0pt,5.0pt) ++(-5.0pt,0) -- ++(5.0pt,-5.0pt);
\draw [color=black] (6.15,2.18)-- ++(-2.5pt,-2.5pt) -- ++(5.0pt,5.0pt) ++(-5.0pt,0) -- ++(5.0pt,-5.0pt);
\draw [color=black] (4.66,2.64)-- ++(-2.5pt,-2.5pt) -- ++(5.0pt,5.0pt) ++(-5.0pt,0) -- ++(5.0pt,-5.0pt);
\draw [color=black] (1.31,-2.09)-- ++(-2.5pt,-2.5pt) -- ++(5.0pt,5.0pt) ++(-5.0pt,0) -- ++(5.0pt,-5.0pt);
\draw [color=black] (1.28,0.53)-- ++(-2.5pt,-2.5pt) -- ++(5.0pt,5.0pt) ++(-5.0pt,0) -- ++(5.0pt,-5.0pt);
\draw [color=black] (4.74,-4.12)-- ++(-2.5pt,-2.5pt) -- ++(5.0pt,5.0pt) ++(-5.0pt,0) -- ++(5.0pt,-5.0pt);
\draw [color=black] (2.24,-3.34)-- ++(-2.5pt,-2.5pt) -- ++(5.0pt,5.0pt) ++(-5.0pt,0) -- ++(5.0pt,-5.0pt);
\draw [color=black] (7.73,-1.48)-- ++(-2.5pt,-2.5pt) -- ++(5.0pt,5.0pt) ++(-5.0pt,0) -- ++(5.0pt,-5.0pt);
\end{scriptsize}
\end{tikzpicture}
\caption{Illustration of the points $x_i^{\gamma,r}$ and arcs $I_i^{\gamma,r}$ in the proof of theorem~\ref{thm:one}. In this example $m_\gamma=2$, $n_\gamma=5$, and $r\approx0.6$. The points $x_i,\dots,x_{10}$ on $S_r$ are denoted by crosses and the endpoints of the arcs $I_i^{\gamma,r}$ by dots. For clarity, only one arc and few points are labeled. The singleton $E$ is highlighted at the boundary $\partial D$. The choice of the arcs $I_i^{\gamma,r}$ presented here is only one of many possibilities.}%
\label{fig:proof-one}%
%\end{center}
\end{figure}
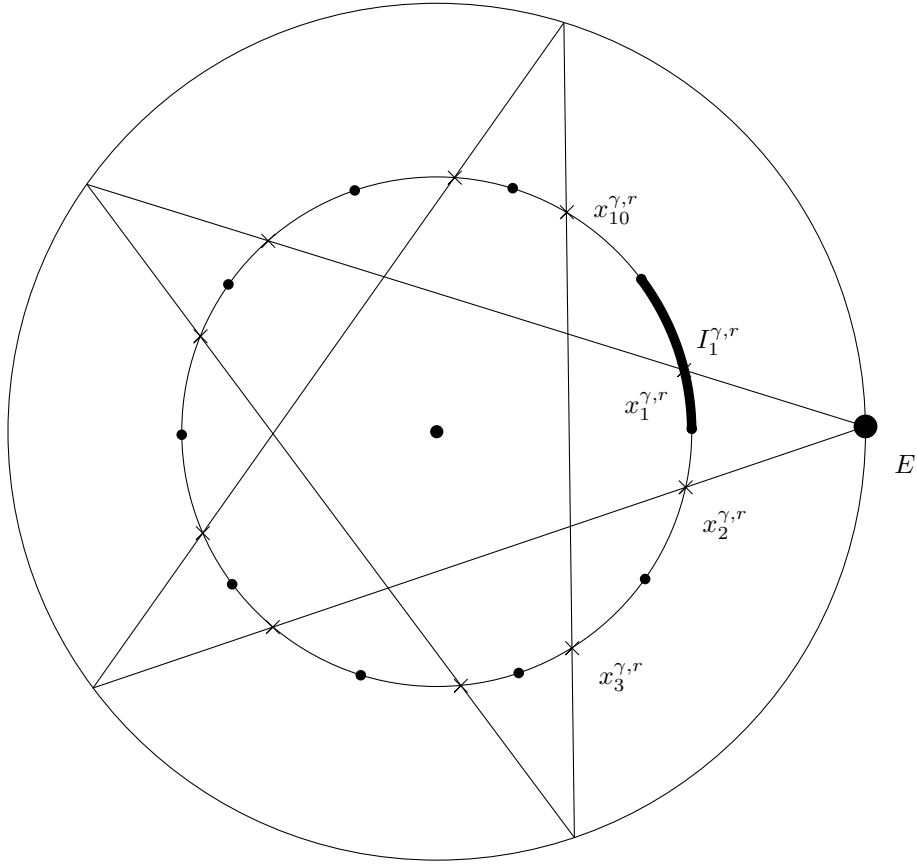

Let $\eps>0$. Since $f$ is continuous in the compact set $\bar{D}$, it is uniformly continuous and there is $\delta>0$ such that $\abs{f(x)-f(y)}<\eps$ whenever $\abs{x-y}<\delta$.
%If we choose $N\in\N$ with $\pi r/N<\delta$, then
If we choose $\gamma$ so that $n_\gamma>\pi/\delta$, then $\pi r/n_\gamma<\delta$ for all $r\in(z_\gamma,1)$, whence for almost every $r\in(z_\gamma,1)$ we find
\begin{equation}
%\label{eq:}
\abs{\fint_{I_i^{\gamma,r}}f\der\h^1-f(x_i^{\gamma,r})}<\eps,
\end{equation}
which yields
\begin{equation}
%\label{eq:}
\begin{split}
\abs{a_0(r)-\frac{1}{2n_\gamma}\sum_{y\in S_r\cap\Tr(\gamma)}f(y)}
&=
\abs{\frac{1}{2\pi r}\sum_{i=1}^{2n_\gamma}\int_{I_i^{\gamma,r}}f\der\h^1-\frac{1}{2n_\gamma}\sum_{i=1}^{2n_\gamma}f(x_i^{\gamma,r})}
\\&=
\abs{\frac{1}{2n_\gamma}\sum_{i=1}^{2n_\gamma}\left(\fint_{I_i^{\gamma,r}}f\der\h^1-f(x_i^{\gamma,r})\right)}
\\&\leq
\frac{1}{2n_\gamma}\sum_{i=1}^{2n_\gamma}\eps
=\eps
\end{split}
\end{equation}
and
\begin{equation}
\label{eq:Gg-estimate}
\begin{split}
\abs{\brt f(\gamma)-g(z_\gamma)}
&=
\abs{\int_{z_\gamma}^1 p_{z_\gamma}(r) \frac{1}{2n_\gamma}\sum_{y\in S_r\cap\Tr(\gamma)} f(y)\der r-\int_z^1 p_z(r) a_0(r) \der r}
\\&\leq
\int_{z_\gamma}^1 p_{z_\gamma}(r)\abs{\frac{1}{2n_\gamma}\sum_{y\in S_r\cap\Tr(\gamma)} f(y)-a_0(r)}\der r
\\&\leq
\int_{z_\gamma}^1 p_{z_\gamma}(r)\eps\der r
=\eps,
\end{split}
\end{equation}
since $\int_z^1p_z(r)\der r=1$.

Comparing \eqref{eq:g-def} with the zeroth Abel transform defined in~\eqref{eq:Abel0}, we have
\begin{equation}
\label{eq:g-Abel}
g(z) = \frac{\A_0a_0(z)}{2\sqrt{1-z^2}}.
\end{equation}
Since $f$ is continuous by assumption, also $a_0$ is continuous. The Abel transform preserves continuity (lemma~\ref{lma:Abel}), so $g$ is continuous on $[0,1)$. For any $r\in[0,1)$ we can find a sequence $(\gamma_i)_{i=1}^\infty$ such that $\gamma_i\in\Gamma_E^i$ and $z_{\gamma_i}\to r$ as $i\to\infty$.

Let $\eps>0$. By the estimate~\eqref{eq:Gg-estimate} we can find $N=N_\eps\in\N$ such that $\abs{\brt f(\gamma_i)-g(z_{\gamma_i})}\leq\frac{1}{2}\eps$ whenever $i\geq N$. Since $z_{\gamma_i}\to r$ and $g$ is continuous, we can choose $N$ so that also $\abs{g(z_{\gamma_i})-g(r)}\leq\frac{1}{2}\eps$ whenever $i\geq N$. By these estimates we have $\abs{\brt f(\gamma_i)-g(r)}\leq\eps$ for $i\geq N$, whence
\begin{equation}
\label{eq:Gg-limit}
\lim_{i\to\infty}\brt f(\gamma_i)=g(r)
\end{equation}
for any $r\in[0,1)$ independently of the choice of the sequence $(\gamma_i)$.
Since each $\brt f(\gamma_i)$ is known ($\brt f$ is known), we may recover the function $g$ from the data.

By~\eqref{eq:g-Abel} this implies that we can recover $\A_0a_0$ on $[0,1)$. Since $\A_0a_0$ is continuous, we can recover it on the whole interval $[0,1]$. The Abel transform $\A_0$ is injective on $C([0,1])$ (to which space $a_0$ belongs by assumption) and has an explicit inversion formula (lemma~\ref{lma:Abel} below), allowing us to finally reconstruct the function $a_0$.
\end{proof}

\begin{proof}[Proof of corollary~\ref{cor:one}]
The geometry of the ball $\bar{B}^n$ is such that any broken ray $\gamma$ lays in the intersection of a two dimensional subspace of $\R^n$ and the ball. One may pick any plane containing the singleton and the origin and reduce the problem to singleton tomography in $\bar{B}^2$. The result then follows trivially from theorem~\ref{thm:one}.
\end{proof}

\section{Integral transforms}
\label{sec:int}

%\subsection{Fourier decomposition}

We write our function $f\in\CDini$ as a Fourier series in the angular variable. The following lemma describes the sense in which the series represents the function. The lemma is true also for the complex Fourier series.

\begin{lemma}
\label{lma:Fourier}
%If a function $f:\bar{D}\to\R$ is continuous and satisfies the Dini condition in some annulus centered at the origin, then its Fourier series converges uniformly to $f$ in this annulus:
The Fourier series of any function $f\in\CDini$ converges to $f$ uniformly in the disk:
There are continuous functions $a_k:[0,1]\to\R$, $k\geq0$, and $b_k:[0,1]\to\R$, $k\geq1$, such that the sequence of functions $(f_K)_K$,
\begin{equation}
%\label{eq:}
f_K(r,\theta) = a_0(r)+\sum_{k=1}^K (a_k(r)\cos(k\theta)+b_k(r)\sin(k\theta)),
\end{equation}
converges uniformly to $f$ in the disk as $K\to\infty$. These functions $a_k$ and $b_k$ admit a common (\DL) modulus of continuity and satisfy $a_k(0)=b_k(0)=0$ when $k\geq1$.
%
%In particular, if $f\in\CDini$ (the annulus is the whole $\bar{D}$), then the uniform convergence holds in the whole disk.
\end{lemma}
\begin{proof}
%We only consider the case $f\in\CDini$; the general case is an obvious generalization.
We set for every $r\in[0,1]$
\begin{equation}
%\label{eq:}
a_0(r) = \frac{1}{2\pi}\int_0^{2\pi}f(r,\theta)\der\theta
\end{equation}
and for $k\geq1$
\begin{equation}
%\label{eq:}
a_k(r) = \frac{1}{\pi}\int_0^{2\pi}f(r,\theta)\cos(k\theta)\der\theta
\end{equation}
and
\begin{equation}
%\label{eq:}
b_k(r) = \frac{1}{\pi}\int_0^{2\pi}f(r,\theta)\sin(k\theta)\der\theta.
\end{equation}
These functions manifestly have the property that $a_k(0)=b_k(0)=0$ when $k\geq1$.

%Since $f$ is continuous in the compact set $\bar{D}$, it is uniformly continuous and admits a modulus of continuity.

The functions $f(r,\cdot)$, $r\in(0,1]$, admit a common modulus of continuity which satisfies the \DL\ condition.
(Clearly the functions $a_k$ and $b_k$ admit essentially the same modulus of continuity.)
The Fourier series of a uniformly continuous function satisfying the \DL\ condition converges uniformly and the rate of convergence depends only on the modulus of continuity~\cite[Theorem~10.3 of Chapter~2]{book-zygmund}, whence each $f_K\to f$ uniformly.
%\cite[Section~2.1 of Chapter~1]{book-cha}
\end{proof}

Let $f\in\CDini$. Using the notations of the above proof, $f_K\to f$ uniformly and $\brt$ is continuous in the corresponding topology, whence $\brt f_K\to\brt f$ uniformly. This allows us to study the Fourier components $a_k$ and $b_k$ of $f$ one by one; this will be done in lemma~\ref{lma:calc} and its corollary.

%\subsection{Generalized Abel transform}

What we will be able to construct from the broken ray transform $\brt f$ is some integral transforms of its Fourier components $a_k$ and $b_k$ as defined in lemma~\ref{lma:Fourier}. Since the case $k=0$ reduces to the classical Abel transform, we call these transforms generalized Abel transforms. These transforms appear implicitly in Cormack's approach~\cite{cormack} to the Radon transform in the plane.

\begin{definition}
\label{def:Abel}
For $k=1,2,\dots$ we define the \emph{$k$th generalized Abel transform} as a map $\A_k:C_0([0,1])\to C([0,1])$ such that
\begin{equation}
\label{eq:Abel}
\A_k f(z) = 2\int_z^1 T_k(z/y) \frac{f(y)}{\sqrt{1-(z/y)^2}}\der y,
\end{equation}
where $T_k$ is the $k$th Chebyshev polynomial and $C_0([0,1])=\{f\in C([0,1]):f(0)=0\}$.

For $k=0$ we define the \emph{zeroth generalized Abel transform} (the Abel transform up to a change of variable) as a map $\A_0:C([0,1])\to C([0,1])$ such that
\begin{equation}
\label{eq:Abel0}
\A_0 f(z) = 2\int_z^1 \frac{f(y)}{\sqrt{1-(z/y)^2}}\der y.
\end{equation}
Since $T_0\equiv1$, \eqref{eq:Abel} is correct also in the case $k=0$.
\end{definition}

We summarise the main properties of these transforms. Some of the properties are based on those of the Radon transform in the plane, which are covered e.g. in~\cite{book-helgason}.

\begin{lemma}
\label{lma:Abel}
The generalized Abel transforms defined in definition~\ref{def:Abel} have the following properties with any $k\geq0$:
\begin{enumerate}
	\item $\A_k$ indeed maps $C_0([0,1])$ ($C([0,1])$ for $k=0$) into $C([0,1])$ as defined.
	\item $\A_k$ is linear.
	\item The transform is continuous and $\aabs{\A_k}\leq2$ when $C([0,1])$ is equipped with the supremum norm.\footnote{Using H\"older's inequality it is not difficult to show that $\A_k:L^p([0,1])\to L^q([0,1])$ is continuous for any $p\in(2,\infty]$ and $q\in[1,\infty]$.}
	\item $\A_k$ is injective.
	\item We have the explicit inversion formula
		\begin{equation}
		%\label{eq:}
		f(y)=-\frac{1}{\pi}\frac{\der}{\der y} \int_y^1 T_k(z/y) \frac{\A_k f(z)}{z\sqrt{(z/y)^2-1}}\der z.
		\end{equation}
\end{enumerate}
\end{lemma}
\begin{proof}
(1) For $f\in C([0,1])$ let $g_K(r,\theta)=f(r)\cos(k\theta)$. Also, let $\radon$ be the Radon transform parametrized so that
\begin{equation}
%\label{eq:}
\radon g_k(\rho,\phi)=\int_{L_{\rho,\phi}}g_k\der\h^1,
\end{equation}
where $L_{\rho,\phi}=\{x\in\R^2:x_1\cos\phi+x_2\sin\phi=\rho\}$.

A simple calculation 
%\footnote{The Chebyshev polynomials arise due to the property $T_k(x)=\cos(k\arccos(x))$.}
yields
\begin{equation}
%\label{eq:}
\radon g_k(\rho,\phi) = \A_k f(\rho) \cos(k\phi).
\end{equation}
Since $f(0)=0$ if $k>0$, $g_k$ is continuous, and so is $\radon g_k$. Therefore the above equation shows that $\A_k f$ is continuous.

(2) Trivial.

(3) The result follows from the estimate $\abs{T_k(x)}\leq1$ and the observation
\begin{equation}
%\label{eq:}
\int_z^1 \frac{1}{\sqrt{1-(z/y)^2}}\der y=\sqrt{1-z^2}.
\end{equation}

(4) The Radon transform is injective on continuous functions, so the observation made in part (1) of the proof shows that $\A_k$ is injective. We also have the inversion formula
\begin{equation}
%\label{eq:}
f(r)
=
\radon^{-1}[(\rho,\phi)\mapsto \cos(k\phi)\A_k f(\rho)](r,0)
\end{equation}
using the inverse Radon transform.

(5) See Eqs.~(10) and~(18) in~\cite{cormack}. The treatment therein only requires that $f$ is continuous.
\end{proof}

\section{Tools for open set tomography}
\label{sec:tools}

\subsection{Explicit form for the broken ray transform}

Using the Fourier series representation given in lemma~\ref{lma:Fourier} and the generalized Abel transforms we can explicitly calculate the broken ray transform in terms of the Fourier components. Lemma~\ref{lma:calc} gives the basic result for each Fourier component, and its corollary lists the relevant special cases we will need.

To simplify further notation, we define the coefficient
\begin{equation}
%\label{eq:}
S_k(\gamma)
=
\begin{dcases}
\frac{\sin(k(\kappa_\gamma-\iota_\gamma)/2)}{\sin(k\alpha_\gamma/2)} & \text{when } k\alpha_\gamma\notin2\pi\Z\\
n_\gamma(-1)^{(n_\gamma+1)\frac{k\alpha_\gamma}{2\pi}+km_\gamma} & \text{when } k\alpha_\gamma\in2\pi\Z
\end{dcases}
\end{equation}
for all $k=0,1,\dots$ and $\gamma\in\Gamma_E$. This coefficient has the following properties.

\begin{lemma}
\label{lma:S_k}
For any $k$ and $\gamma$ we have:
\begin{enumerate}
	\item $\abs{S_k(\gamma)}\leq n_\gamma$.
	\item The sum identities
%	\begin{equation}%\label{eq:}
%	\sum_{l=1}^{n_\gamma}\cos(k(l-1/2)\alpha_\gamma)=S_k(\gamma)\cos(k(\kappa_\gamma-\iota_\gamma)/2)
%	\end{equation}
%	and
%	\begin{equation}%\label{eq:}
%	\sum_{l=1}^{n_\gamma}\sin(k(l-1/2)\alpha_\gamma)=S_k(\gamma)\sin(k(\kappa_\gamma-\iota_\gamma)/2)
%	\end{equation}
	\begin{equation}%\label{eq:}
	\begin{split}
	\sum_{l=1}^{n_\gamma}\cos(k(l-1/2)\alpha_\gamma)&=S_k(\gamma)\cos(k(\kappa_\gamma-\iota_\gamma)/2)\quad\text{and}\\
	\sum_{l=1}^{n_\gamma}\sin(k(l-1/2)\alpha_\gamma)&=S_k(\gamma)\sin(k(\kappa_\gamma-\iota_\gamma)/2)
	\end{split}
	\end{equation}
	hold.
	\item If $(\kappa_\gamma-\iota_\gamma)/\pi$ or $\alpha_\gamma/\pi$ is irrational, $S_k(\gamma)\neq0$.
\end{enumerate}
\end{lemma}
\begin{proof}
(1) By \eqref{eq:trajectory}
\begin{equation}
%\label{eq:}
\sin(k(\kappa_\gamma-\iota_\gamma)/2)
=
\sin(kn_\gamma\alpha_\gamma/2-km_\gamma\pi)
=
(-1)^{km_\gamma}\sin(kn_\gamma\alpha_\gamma/2),
\end{equation}
so it suffices to show that the function
\begin{equation}
%\label{eq:}
f(x)=\frac{\sin(nx)}{\sin(x)}
\end{equation}
satisfies $\abs{f(x)}\leq n$ for all $x\in\R\setminus\pi\Z$ when $n\in\N$ and $n\geq2$.

If we define $f(m\pi)=n(-1)^{(n-1)m}$ for all $m\in\Z$, the function $f$ becomes real analytic on the real line. In particular, we have $f(0)=n$.

By symmetry, it is enough to show that $\abs{f(x)}\leq n$ for all $x\in[0,\pi/2]$. To this end we wish to show that $f(0)=n$, $f'\neq0$ on $(0,\pi/2n)$, and $\abs{f(x)}\leq n$ when $x\in[\pi/2n,\pi/2]$.

Indeed, if $f'(x)=0$, then $n\tan(x)=\tan(nx)$, but this equation has no solutions on $(0,\pi/2n)$. And since the sine function is concave on $[0,\pi/2]$, we have $\sin(\pi/2n)\geq 1/n$, whence $1/\sin(x)\leq n$ for all $x\in[\pi/2n,\pi/2]$. This yields the desired estimate.

(2) Taking real and imaginary parts of the geometric sum
\begin{equation}
%\label{eq:}
\sum_{l=1}^{n}e^{ik(l-1/2)\alpha}
=
\begin{dcases}
\frac{\frac{i}{2}(1-e^{ik\beta})}{\sin(k\alpha/2)} & \text{when } e^{ik\alpha}\neq1\\
n e^{ik\alpha/2} & \text{when } e^{ik\alpha}=1,
\end{dcases}
\end{equation}
when $\alpha$ and $\beta$ satisfy $n\alpha=\beta+2\pi m$ for an integer $m$, immediately gives the result.

(3) By \eqref{eq:trajectory} $(\kappa_\gamma-\iota_\gamma)/\pi$ is irrational if and only if $\alpha_\gamma/\pi$ is. The conclusion follows thus immediately from the definition of $S_k(\gamma)$.
\end{proof}

\begin{lemma}
\label{lma:calc}
Suppose $\gamma\in\Gamma_E$ is such that $\iota_\gamma=0$.
%Olkoot $a,b\in S_1$ siten, että $a$ on kulmassa nolla ja $b$ kulmassa $\beta$.\beta\to\kappa_\gamma
If $f(r,\theta) = a_k(r) \cos(k\theta) + b_k(r) \sin(k\theta)$, where $k=0,1,\dots$ and the functions $a_k$ and $b_k$ are continuous, then
\begin{equation}
%\label{eq:}
\nbrt f(\gamma)
=
S_k(\gamma)
(\cos(k\kappa_\gamma/2)\A_k a_k(z_\gamma)+\sin(k\kappa_\gamma/2)\A_k b_k(z_\gamma)).
\end{equation}
\end{lemma}
\begin{proof}
Let us only consider the case $b_k=0$; the calculations for the $b_k$ term are essentially the same.

Fix some $\gamma\in\Gamma_E$ with $\iota_\gamma=0$. For any radius $r\in[0,1)$ and direction $\phi\in\partial D$ we define the corresponding secant $H_{r,\phi}=\{y\in D:\phi\cdot y=r\}$. The trajectory $\gamma$ is composed of $n_\gamma$ such secants. Each of these secants corresponds to the radius $r=z_\gamma$. We identify the directions $\phi$ with angles and label them by $\phi_l=(l-1/2)\alpha_\gamma$, where $l=1,2,\dots,n_\gamma$.

After a change of the variable, we find
\begin{equation}%\label{eq:}
\int_{H_{z_\gamma,\phi_l}} f \der\h^1
=
\int_{-\alpha_\gamma/2}^{\alpha_\gamma/2} a_k\left(\frac{\cos(\alpha_\gamma/2)}{\cos(\phi)}\right) \cos(k(\phi_l+\phi)) \frac{\cos(\alpha_\gamma/2)}{\cos^2(\phi)} \der\phi.
\end{equation}
Since
\begin{equation}
%\label{eq:}
\cos(k(\phi_l+\phi)) = \cos(k\phi_l)\cos(k\phi) - \sin(k\phi_l)\sin(k\phi)
\end{equation}
and the second term vanishes in the integral, another change of variable $r=z_\gamma/\cos(\phi)$ gives
\begin{equation}
%\label{eq:}
\begin{split}
\int_{H_{z_\gamma,\phi_l}} f \der\h^1
&=
\cos(k\phi_l)
%\\&\times
\int_{-\alpha_\gamma/2}^{\alpha_\gamma/2} a_k\left(\frac{\cos(\alpha_\gamma/2)}{\cos(\phi)}\right) \cos(k\phi) \frac{\cos(\alpha_\gamma/2)}{\cos^2(\phi)} \der\phi
\\&=
2\cos(k\phi_l)
%\\&\times
\int_{0}^{\alpha_\gamma/2} a_k\left(\frac{z_\gamma}{\cos(\phi)}\right) T_k(\cos(\phi)) \frac{1}{\sin(\phi)} \frac{z_\gamma\sin(\phi)}{\cos^2(\phi)} \der\phi
\\&=
2\cos(k\phi_l)
%\\&\times
\int_{z_\gamma}^{1} a_k(r) T_k(z_\gamma/r) \frac{1}{\sqrt{1-(z_\gamma/r)^2}} \der r
\\&=
\cos(k\phi_l)\A_k a_k(z_\gamma).
\end{split}
\end{equation}
The secants $H_{z_\gamma,\phi_l}$, $l=1,\dots,n_\gamma$, intersect only at a finite number of points and $\Tr(\gamma)=\bigcup_{l=1}^{n_\gamma}H_{z_\gamma,\phi_l}$, so
\begin{equation}
%\label{eq:}
\nbrt f(\gamma)
=
\sum_{l=1}^{n_\gamma}\cos(k\phi_l)\A_k a_k(z_\gamma).
\end{equation}
Using lemma~\ref{lma:S_k} to evaluate the sum, we arrive at the claim.
\end{proof}

\begin{corollary}
\label{cor:calc}
Let $f$ be as in lemma~\ref{lma:calc}. If $0\in E$ and $\gamma\in\Gamma_E$ is such that $\iota_\gamma+\kappa_\gamma=0$, then
\begin{equation}%\label{eq:}
\nbrt f(\gamma)
=
S_k(\gamma)
\A_k a_k(z_\gamma).
\end{equation}
In particular, if $E=\{0\}$ so that $\iota_\gamma=\kappa_\gamma=0$, for injectivity it is sufficient to consider only those $\gamma$ for which $\gcd(n_\gamma,m_\gamma)=1$. In this case
\begin{equation}%\label{eq:}
\nbrt f(\gamma) = n_\gamma (-1)^{km_\gamma/n_\gamma} \A_k a_k(z_\gamma)
\end{equation}
if $k/n_\gamma$ is an integer and $\nbrt f(\gamma)=0$ otherwise.
\end{corollary}
\begin{proof}
The function $f$ may be written as a component of a Fourier series in the variable $\tilde{\theta}=\theta-\iota_\gamma$:
\begin{equation}
%\label{eq:}
f(r,\tilde{\theta}) = \tilde{a}_k(r)\cos(k\tilde{\theta})+\tilde{b}_k(r)\sin(k\tilde{\theta}).
\end{equation}
This corresponds to choosing coordinates so that $\iota_\gamma=0$, so that the problem reduces to the situation in lemma~\ref{lma:calc}. Due to elementary trigonometric identities the coefficients of this Fourier series relate to the ones in lemma~\ref{lma:Fourier} so that
%\begin{equation}%\label{eq:}
%\begin{split}
%\tilde{a}_k(r) &= a_k(r) \cos(k\iota_\gamma) + b_k(r) \sin(k\iota_\gamma) \quad\text{and}\\
%\tilde{b}_k(r) &= -a_k(r) \sin(k\iota_\gamma) + b_k(r) \cos(k\iota_\gamma).
%\end{split}
%\end{equation}
\begin{equation}%\label{eq:}
\begin{split}
a_k(r) &= \tilde{a}_k(r) \cos(k\iota_\gamma) - \tilde{b}_k(r) \sin(k\iota_\gamma) \quad\text{and}\\
b_k(r) &= \tilde{a}_k(r) \sin(k\iota_\gamma) + \tilde{b}_k(r) \cos(k\iota_\gamma).
\end{split}
\end{equation}
Lemma~\ref{lma:calc} indicates now that
\begin{equation}
%\label{eq:}
\begin{split}
\nbrt f(\gamma)
&=
S_k(\gamma)
(\cos(k\kappa_\gamma)\A_k \tilde{a}_k(z_\gamma)+\sin(k\kappa_\gamma)\A_k \tilde{b}_k(z_\gamma))
\\&=
S_k(\gamma)
(\cos(k\iota_\gamma)\A_k \tilde{a}_k(z_\gamma)-\sin(k\iota_\gamma)\A_k \tilde{b}_k(z_\gamma))
\\&=
S_k(\gamma)\A_k a_k(z_\gamma),
\end{split}
\end{equation}
which is the first part of the claim.

Let us now turn to the case $E=\{0\}$. Given $\gamma$, let $\tilde{\gamma}$ be such that
$n_{\tilde{\gamma}}=n_\gamma/\gcd(n_\gamma,m_\gamma)$,
$m_{\tilde{\gamma}}=m_\gamma/\gcd(n_\gamma,m_\gamma)$,
$\alpha_{\tilde{\gamma}}=\alpha_\gamma$,
and $\Tr(\tilde{\gamma})=\Tr(\gamma)$.
It is geometrically obvious that $\gamma$ contains $\gcd(n_\gamma,m_\gamma)$ copies of $\tilde{\gamma}$ and so $\brt f(\gamma)=\brt f(\tilde{\gamma})$.
%Thus considering those $\gamma\in\Gamma_{\{0\}}$ with coprime $n_\gamma$ and $m_\gamma$ is enough.
Thus it suffices to consider those $\gamma\in\Gamma_{\{0\}}$ with coprime $n_\gamma$ and $m_\gamma$.

Now that $n_\gamma\alpha_\gamma=2\pi m_\gamma$ and $\gcd(n_\gamma,m_\gamma)=1$, $k\alpha_\gamma\in2\pi\Z$ if and only if $k/n_\gamma$ is an integer. Thus the definition of $S_k(\gamma)$ together with the result $\nbrt f(\gamma)=S_k(\gamma)\A_k a_k(z_\gamma)$ gives the second part of the claim.
\end{proof}

Note that the functions $b_k$ do not appear in the above formulas for $\brt f(\gamma)$. This is due to the assumption that the broken ray is symmetric with respect to the normal to $\partial D$ at angle $0$, which causes functions antisymmetric with respect to this line to integrate to zero over the trajectory.

\subsection{Quasianalytic functions}

The class of functions for which our argument shows the uniqueness result of theorem~\ref{thm:open} consists of functions uniformly quasianalytic in the angular variable as defined below. There is, however, no implication that this is the optimal class in which the theorem should hold.

\begin{definition}
\label{def:qa}
Let $M=(M_n)_{n=0}^\infty$ be an increasing logarithmically convex sequence\footnote{The sequence $(M_n)_{n=0}^\infty$ is logarithmically convex if the sequence $(\log M_n)_{n=0}^\infty$ is convex.} of strictly positive real numbers. A sequence $(a_k)_{k=1}^\infty$ of real numbers is defined to be in the class $S^\#(M)$ if there is a constant $R>0$ such that for each $n$ it satisfies%\footnote{We use the convention that $0^0=1$, but it makes no real difference, since it is only the asymptotical behaviour of the sequence that matters.}
\begin{equation}
%\label{eq:}
\sum_{k=0}^\infty k^{2n}\abs{a_k}^2 \leq M_n R^n.
\end{equation}
The class $S^\#(M)$ is called a \emph{quasianalytic class of sequences} if the defining sequence $M$ satisfies
\begin{equation}
%\label{eq:}
\sum_{n=0}^\infty \sqrt{\frac{M_n}{M_{n+1}}}=\infty.
\end{equation}
A function $f:\R\to\R$ with a period $2\pi$ belongs to the class $C^\#(M)$ if it can be written as a uniformly convergent Fourier series as
\begin{equation}
%\label{eq:}
f(x) = \sum_{k=0}^\infty (a_k\cos(kx)+b_k\sin(kx))
\end{equation}
such that the sequences $(a_k)_{k=1}^\infty$ and $(b_k)_{k=1}^\infty$ are in the class $S^\#(M)$. If the class $S^\#(M)$ is quasianalytic, then $C^\#(M)$ is a \emph{quasianalytic class of functions}. Functions in such a class are \emph{quasianalytic functions}.

If a continuous function $f:\bar{D}\to\R$ is written as a Fourier series as in lemma~\ref{lma:Fourier}, it is \emph{uniformly quasianalytic in the angular variable} if for each $R\in(0,1]$ the collection of functions $\{f(r,\cdot):r\geq R\}$ belong to the same quasianalytic class.

In dimensions higher than 2, we define a function $f:\bar{B}^n\to\R$, $n\geq3$, to be uniformly quasianalytic in the angular variables if its restriction to any two dimensional plane intersecting the origin is uniformly quasianalytic in the sense defined in two dimensions.%\footnote{TODO: Could this be stated in terms of the coefficients of the function written as a series of multidimensional spherical harmonics? How can one control the surface gradients of spherical harmonics with the indices of the spherical harmonic? How do the Fourier coefficients of the restriction to a plane relate to the spherical harmonic coefficients?}
\end{definition}

The class $C^\#(M)$ of $2\pi$-periodic smooth functions is a quasianalytic class of functions in the sense defined above if and only if it is so in the classical sense: if $f$ is in a quasianalytic class and the function and all of its derivatives vanish at any point, then the function is identically zero. For this result and more discussion on periodic quasianalytic functions we refer to~\cite[Section~V.2]{book-katznelson}. %\footnote{Our $M_n$ and $R$ are square roots of the ones of Katznelson.}.
We only remark here that quasianalytic functions are smooth and all derivatives of the Fourier series (which can be taken term by term) converge uniformly to the corresponding derivative of the function.

Angular differentiation of a function uniformly quasianalytic in the angular variable may also be carried out term by term and the series converges uniformly, provided that the angular derivatives of all orders satisfy the \DL\ condition. This is why we make the assumption in theorem~\ref{thm:open}.

%If we had $\sum_kk^n\omega_{a_k|_{[r,1]}}(\delta)=\order(\log^{-1}\delta^{-1})$ for all $n\in\N,r\in(0,1)$, then $\partial_\ang^n f$ would be \DL\ (in annuli) and uqa for all $n$ if $f$ is. It follows from this requirement that $\sum_kk^n\aabs{a_k|_{[r,1]}}<\infty\forall n$; if we had also a bound for $\sum_kk^n\aabs{a_k|_{[r,1]}}^2$, uqa would follow.
%
%Also $\omega_{a_k|_{[r,1]}}(\delta)\leq C_\delta\aabs{a_k|_{[r,1]}}$ with $C_\delta=\order(\log^{-1}\delta^{-1})$ with $f$ uqa would be enough.

%The following lemma follows immediately from the definition.
%
%\begin{lemma}
%\label{lma:qa-Dini}
%A uniformly quasianalytic function is Dini-continuous in any annulus $\{x:\in\R^2:r<\abs{x}\leq1\}$ where $r>0$.%Not true even if each a_k is continuous?
%\end{lemma}

The important property of quasianalytic functions that motivates the use of this class is the following lemma.

\begin{lemma}
\label{lma:qa}
Let $(x_k)_{k=0}^\infty$ be a bounded sequence and $z\in(0,1)$.
If a uniformly quasianalytic function $f$ is written as a Fourier series as in lemma~\ref{lma:Fourier} and
\begin{equation}
%\label{eq:}
\sum_{k=0}^\infty (-k^2)^n x_k \A_k a_k(z) = 0
\end{equation}
for all $n=0,1,\dots$, then $x_k\A_k a_k(z)=0$ for each $k$.
\end{lemma}
\begin{proof}
The functions $a_k(r,\cdot)$, $r\geq z$, belong to the same quasianalytic class $C^\#(M)$. By lemma~\ref{lma:Abel}(3) we have $(\A_k a_k(z))_{k=0}^\infty\in S^\#(M)$.

Define a function $g:\R\to\R$ so that
\begin{equation}
%\label{eq:}
g(x)=\sum_{k=0}^\infty x_k \A_k a_k(z) \cos(kx).
\end{equation}
Now $g$ is quasianalytic and obviously its odd derivatives vanish at zero. By assumption also the function itself and its even derivatives vanish at zero, so $g$ must be identically zero. This proves the lemma.
\end{proof}

\section{Tomography from a singleton}
\label{sec:one}

%We now prove theorem~\ref{thm:one}. The proof also gives an explicit reconstruction method. This result is discussed in more detail after the proof.

We prove theorem~\ref{thm:one} in the special case when the unknown function $f$ satisfies the \DL\ condition. This result is not needed for the proof of the general case (which was presented in section~\ref{sec:one-pf}), but we present it since it is simple and uses the tools developed for an open set of tomography. Both proofs are constructive and the reconstructions are based on the `long trajectory limit', where $n_\gamma$ tends to infinity.

\begin{proposition}
\label{prop:var-one}
For $f\in\CDini$, $\brt f:\Gamma_E\to\R$ determines the zeroth Fourier component $a_0(r)=\fint_{S(0,r)}f\der\h^1$ when $E$ is a singleton.
\end{proposition}

\begin{proof}
As noted in corollary~\ref{cor:calc}, it is sufficient to consider those $\gamma\in\Gamma_E$ for which $\gcd(n_\gamma,m_\gamma)=1$.
By continuity of $f$ it is enough to reconstruct it outside the boundary $\partial D$.

Let $a_k$ and $b_k$ be as in lemma~\ref{lma:Fourier}. Let $g(r,\theta)=a_0(r)$ and
\begin{equation}
%\label{eq:}
f_K(r,\theta) = \sum_{k=1}^K (a_k(r)\cos(k\theta)+b_k(r)\sin(k\theta)).
\end{equation}
By corollary~\ref{cor:calc} $\brt g(\gamma)=d_\gamma^{-1}\A_0a_0(z_\gamma)$, and using this we show that $\brt f(\gamma)-\brt g(\gamma)\to0$ as $n_\gamma\to\infty$.

Let $\eps>0$. By lemma~\ref{lma:Fourier} there is $K_\eps$ such that $\aabs{f-g-f_{K_\eps}}<\eps$.
By corollary~\ref{cor:calc} $\brt f_{K_\eps}(\gamma)=0$ whenever $\gamma\in \Gamma_E^{K_\eps+1}$, whence
\begin{equation}
%\label{eq:}
\abs{\brt f(\gamma)-\brt g(\gamma)}
=
\abs{\brt(g+f_{K_\eps}-f)(\gamma)-\brt f_{K_\eps}(\gamma)}
\leq
\eps+0.
\end{equation}
This estimate holds for all $\gamma\in \Gamma_E^{K_\eps+1}$.

Fix any $w\in[0,1)$. There is a sequence $(\gamma_i)_{i\in\N}$ such that $n_{\gamma_i}\geq i$, and $z_{\gamma_i}\to w$ (and thus $d_{\gamma_i}\to 2\sqrt{1-w^2}$) as $i\to\infty$. The estimate obtained above gives thus $\lim_{i\to\infty}\brt g(\gamma_i)=\frac{1}{2}(1-w^2)^{-1/2}\A_0a_0(w)$, allowing one to reconstruct the function $\A_0a_0$ on $[0,1)$. Injectivity of the Abel transform (lemma~\ref{lma:Abel}) proves the claim. 
\end{proof}

%We now turn to the general case, where the assumption $f\in\CDini$ of proposition~\ref{prop:var-one} is replaced by $f\in C$ without any assumptions on the modulus of continuity.

%Proof of thm:one moved from here.

By corollary~\ref{cor:calc} and theorem~\ref{thm:one} we have the following for $f\in\CDini$ and a singleton $E$ (chosen to be at angle zero): $\brt f|_{\Gamma_E}=0$ if and only if $a_0=0$ and
\begin{equation}
%\label{eq:}
%\sum_{n_\gamma|k}(-1)^{km_\gamma/n_\gamma}\A_k a_k(z_\gamma)=0
\sum_{n_\gamma|k}(-1)^{km_\gamma/n_\gamma}\A_k a_k(z_\gamma)=0
\end{equation}
for all $\gamma\in\Gamma_E$ with $\gcd(n_\gamma,m_\gamma)=1$. 
%Here the sum is taken over all $k$ that have $n_\gamma$ as a factor.
Here the sum is over all $k$ with $n_\gamma|k$.

First, nothing can be said about the functions $b_k$. Furthermore, there are a number of functions $a_k$ such that $\A_k a_k(z_\gamma)=0$ whenever $n_\gamma|k$; by injectivity $\A_k(C_0([0,1]))$ is an infinite dimensional subspace of $C([0,1])$, and the space $\{f\in\A_k(C_0([0,1])):f(z_\gamma)=0\text{ when }n_\gamma|k\}$ has finite codimension in $\A_k(C_0([0,1]))$ (for each $k$, there are only finitely many $\gamma$ with $n_\gamma|k$) and therefore infinite dimension.
Hence the condition $\brt f|_{\Gamma_E}=0$ leaves great freedom to the functions $a_k$, $k\geq1$, and we can efficiently recover only $a_0$ from $\brt f|_{\Gamma_E}$.

One can only possibly image the unknown function $f$ in the set $\bigcup_{\gamma\in\Gamma_E}\Tr(\gamma)$. When $E$ is countable, this set is dense in the disk but has Hausdorff dimension one, so without continuity assumptions no recovery results can be expected. %In this sense theorem~\ref{thm:one} is optimal.

\section{Tomography from an open set}
\label{sec:open}

We denote by $C^{n,m}=C^n_\rad C^m_\ang$ the space of real valued functions on the punctured disk $\bar{D}^*=\bar{D}\setminus\{0\}$ which have $n$ continuous partial derivatives with respect to the radial and $m$ to the angular variable.
%In analogue, we denote the space of uniformly quasianalytic functions (definition~\ref{def:qa}) satisfying the Dini criterion by $C^{0,\qa}$.
The angular derivative is denoted by $\partial_\ang$. We require no regularity at the origin and correspondingly restrict $\brt f$ to $\Gamma_E^*=\{\gamma\in\Gamma_E:0\notin\Tr(\gamma)\}$ for $f\in C^{n,m}$.
Including or excluding the origin in the domain of the unknown function does not alter the results of theorems~\ref{thm:one} and~\ref{thm:open} besides reconstructibility at the origin.
%For a function in $C^{n,m}$ that extends continuously to the origin the values in $\bar{D}^*$ determine the value at the origin.

%In this notation theorem~\ref{thm:open} states that for $E\subset S_1$ open and $f\in C^{0,\qa}$ the condition $\brt f(\gamma)=0$ for all $\gamma\in\Gamma_E$ implies $f=0$. We begin the proof by a lemma, which also provides a tool to generalize the uniqueness result from $C^{0,\infty}$ to $C^{0,0}=C$, but our proof of theorem~\ref{thm:open} only works for the class $C^{0,\qa} \subsetneq C^{0,\infty}$.

%Let us define a rotation operator $\rot_\phi$, $\phi\in\R$, acting on functions so that $\rot_\phi f(r,\theta)=f(r,\theta+\phi)$.
Let a rotation operator $\rot_\phi$, $\phi\in\R$, act on functions so that $\rot_\phi f(r,\theta)=f(r,\theta+\phi)$.
Let another rotation operator $\rot_\phi$ act on trajectories instead of functions so that $\iota_{\rot_\phi\gamma}=\iota_{\gamma}+\phi$, $\kappa_{\rot_\phi\gamma}=\kappa_{\gamma}+\phi$ and other parameters describing $\gamma$ are unaltered. Geometrically $\rot_\phi$ simply rotates the trajectory by an angle $\phi$.

We can only show the uniqueness result of theorem~\ref{thm:open} for functions uniformly quasianalytic in the angular variable, but if it would hold for functions in $C^{0,\infty}$, the following lemma implies that it would also hold for the class $C^{0,0}$.
%The following lemma allows to generalize the uniqueness result from $C^{0,\infty}$ to $C^{0,0}=C$, but we can only show theorem~\ref{thm:open} for uniformly quasianalytic functions.

\begin{lemma}
\label{lma:rot}
For $E\subset\partial D$ open and any $n$, the following hold true:
\begin{enumerate}
	\item If the uniqueness result of theorem~\ref{thm:open} would hold for the class $C^{n,\infty}$, then it would also hold for $C^{n,0}$.
	\item If $f\in C^{n,m}$ and $\brt f(\gamma)=0$ for all $\gamma\in\Gamma_E^*$, then also $\brt\partial_\ang^i f(\gamma)=0$ for all $\gamma\in\Gamma_E^*$ and $i\in\N$, $i\leq m$.
	\item Suppose $f\in C^{0,0}$ and $\brt f(\gamma)=0$ for all $\gamma\in\Gamma_E^*$. Fix $\gamma\in\Gamma_E^*$ and $\delta>0$ such that $\rot_\phi\gamma\in\Gamma_E^*$ whenever $\abs{\phi}<\delta$. Then $\brt \rot_\phi f(\gamma)=0$ for all $\abs{\phi}<\delta$.
\end{enumerate}
%These conclusions hold true if the requirement that the continuous derivatives satisfy the Dini conditions is removed.
\end{lemma}
\begin{proof}
Each part is based on the following observation.  The change of the integral of $f$ over a trajectory is the same if the function is rotated clockwise as if the trajectory is rotated counterclockwise by the same angle $\phi$:
\begin{equation}
\label{eq:rot}
\brt f(\rot_{-\phi}\gamma)
=
\brt \rot_\phi f(\gamma).
\end{equation}

(1) Let us denote a shrinked set of tomography by
\begin{equation}
%\label{eq:}
E^\delta = \{a\in E:d(x,\partial D\setminus E)>\delta\}.
\end{equation}
For a sufficiently small $\delta>0$ the set $E^\delta$ is nonempty and open.

%Let $\eta:\R\to\R$ be a mollifier satisfying the following conditions: $\spt\eta\subset[-1,1]$, $\eta\in C^\infty$, and $\int_{-\infty}^\infty\eta\neq0$. Let us denote its scaling by $\eta_\eps(x)=\eps^{-1}\eta(x/\eps)$ for all $\eps\in(0,1]$.
Let $\eta:\R\to\R$ be the standard mollifier and let $\eta_\eps(x)=\eps^{-1}\eta(x/\eps)$.
%Let $\eps<\delta$; more conditions for the parameter $\eps$ will be given later.
Let $\eps<\delta$.

If $f\in C^{n,0}$ is such that $\brt f|_{\Gamma_E^*}=0$, then $\brt f(\rot_\phi\gamma)=0$ for all $\gamma\in\Gamma_{E^\delta}^*$ and $\abs{\phi}<\delta$. Thus for any $\gamma\in\Gamma_{E^\delta}^*$ we have
\begin{equation}
%\label{eq:}
\begin{split}
0
&=
\int_{-\infty}^\infty \eta_\eps(\phi) \brt f(\rot_{\phi}\gamma) \der\phi
\\&=
\int_{-\infty}^\infty \eta_\eps(\phi) \brt\rot_{-\phi} f(\gamma) \der\phi
\\&=
\brt(\eta_\eps*f)(\gamma),
\end{split}
\end{equation}
where the convolution is with respect to the angular variable of $f$. Therefore $\brt f|_{\Gamma_E^*}=0$ implies that $\brt (\eta_\eps*f)|_{\Gamma_{E^\delta}^*}=0$. If $f\in C^{n,0}$, then $\eta_\eps*f\in C^{n,\infty}$.

If theorem~\ref{thm:open} holds in the class $C^{n,\infty}$, then $\brt (\eta_\eps*f)|_{\Gamma_{E^\delta}^*}=0$ implies that $\eta_\eps*f=0$.
Since this holds for all $\eps\in(0,\delta)$ and $\eta_\eps*f\to f$ uniformly as $\eps\to0$, we find that $f=0$.
%Since for some $\eps\in(0,\delta)$ the map $f\mapsto \eta_\eps*f$ is injective, using such $\eps$ we find that indeed $f=0$.
%Hence it suffices to show that if $\eta_\eps*f=0$, then $f=0$, if $\eps\in(0,\delta)$ is chosen suitably.
%
%We define the functions $e_k(x)=e^{ikx}$ for all $k\in\Z$ and the inner product of two functions $f,g:\R\to\C$ of period $2\pi$ as $\braket{f}{g}=\int_0^{2\pi}f\bar{g}$. A simple calculation shows that $\braket{f*g}{e_k}=\braket{f}{e_k}\braket{g}{e_k}$. The set $\{e_k:k\in\Z\}$ is an orthogonal Hilbert basis~\cite[Theorem~4.2e]{book-baggett} for the space of $L^2$ functions on $\R/2\pi\Z$, so the convolution $f\mapsto \eta_\eps*f$ is injective if and only if $\braket{\eta_\eps}{e_k}\neq0$ for all $k\in\Z$.
%
%We find
%\begin{equation}%\label{eq:}
%\begin{split}
%\braket{\eta_\eps}{e_k}
%&=
%\int_{-\infty}^\infty \eta(x/\eps)e^{-ikx} \der x
%\\&=
%\eps\sqrt{2\pi}\F\eta(k\eps),
%\end{split}
%\end{equation}
%where $\F\eta$ is the Fourier transform of $\eta$. We must choose $\eps$ so that the set $\eps\Z$ contains no zeros of $\F\eta$.
%
%First, the assumption $\int_{-\infty}^\infty\eta\neq0$ implies $\F\eta(0)\neq0$. By the Paley-Wiener theorem~\cite[Theorem~7.22]{book-rudin} $\F\eta$ can be extended to an analytic function in the complex plane. Since it is not identically zero, it can have at most countably infinitely many zeros on the real axis. Therefore there is $c>0$ such that $0\notin\F\eta(c\Q)$; by choosing $\eps\in c\Q$ we indeed have $\braket{\eta_\eps}{e_k}\neq0$ for all $k\in\Z$.

(2) Suppose $m\geq1$ and $f\in C^{n,m}$. It is enough to show the result for the first derivative, from which the argument can be iterated.

Dividing \eqref{eq:rot} by $\phi$ and passing to the limit $\phi\to0$ we obtain
\begin{equation}
%\label{eq:}
\left.\partial_\phi\brt f(\rot_\phi\gamma)\right|_{\phi=0}
=
-
\brt\partial_\ang f(\gamma)
\end{equation}
for all $\gamma\in\Gamma_E^*$.
Since $\brt f$ vanishes on $\Gamma_E^*$, this implies that $\brt\partial_\ang f$ does, too.

(3) Trivial.
\end{proof}

%We remark that part~(1) of the preceeding lemma cannot be easily extended to uniformly quasianalytic functions; there is no nonzero quasianalytic mollifier $\eta:\R\to\R$ with compact support. The following proof relies heavily on quasianalyticity, and so cannot be easily generalized for smooth functions.

\begin{proof}[Proof of theorem~\ref{thm:open}]
We show that $\brt f(\gamma)=0$ for all $\gamma\in\Gamma_E$ implies $f=0$ if $f$ is uniformly quasianalytic in the angular variable and satisfies the \DL\ condition.

%Let us use the functions $e_k(x)=e^{ikx}$ for all $k\in\Z$ as in lemma~\ref{lma:rot}. For simplicity, we write the function $f$ as a complex Fourier series:
%\begin{equation}%\label{eq:}
%f(r,\theta) = \sum_{k\in\Z}c_k(r)e_k(\theta).
%\end{equation}
%For sufficiently small $\phi$ we have
%\begin{equation}%\label{eq:}
%\rot_\phi f(r,\theta) = \sum_{k\in\Z}e^{ik\phi}c_k(r)e_k(\theta).
%\end{equation}
Let $f$ be written as a Fourier series as in lemma~\ref{lma:Fourier}. By theorem~\ref{thm:one} we already know that $a_0=0$. Differentiating term by term, we have
\begin{equation}
%\label{eq:}
\partial_\ang^{2n} f(r,\theta) = \sum_{k=0}^\infty (-k^2)^n (a_k(r)\cos(k\theta)+b_k(r)\sin(k\theta))
\end{equation}
for any $n\in\N$.
We may choose coordinates so that the angle zero lies in $E$. Fix any $\gamma\in\Gamma_E$ such that $\iota_\gamma+\kappa_\gamma=0$; such a trajectory can be found since $E$ is open.

By corollary~\ref{cor:calc}, lemma~\ref{lma:rot} and the assumption of vanishing broken ray transform we have
\begin{equation}
%\label{eq:}
\sum_{k=0}^\infty (-k^2)^n S_k(\gamma) \A_k a_k(z_\gamma) = 0
\end{equation}
for $n=0,1,\dots$. Hence by lemmas~\ref{lma:S_k} and~\ref{lma:qa} we have that $S_k(\gamma)\A_k a_k(z_\gamma)=0$.\footnote{
This conclusion may also be drawn as follows.
%Write $f$ as a complex Fourier series: $f(r,\theta)=\sum_{k\in\Z}a_k(r)e^{ik\theta}$, where the coefficient functions $a_k$ may be complex.
Write $f$ as a complex Fourier series: $f(r,\theta)=\sum_{k\in\Z}a_k(r)e^{i k\theta}$.
%Define the function $h:\R\to\C$ by $h(\phi)=\nbrt\rot_\phi f(\gamma)=\sum_{k\in\Z}e^{ik\phi}S_{\abs{k}}(\gamma)\A_{\abs{k}} a_k(z_\gamma)$.
Define $h:\R\to\C$ by $h(\phi)=\nbrt\rot_\phi f(\gamma)=\sum_{k\in\Z}e^{i k\phi}S_{\abs{k}}(\gamma)\A_{\abs{k}} a_k(z_\gamma)$.
The sequences $(S_k(\gamma)\A_k a_{\pm k}(r))_{k=0}^\infty$ are quasianalytic sequences, whence $h$ is a quasianalytic function.
By part~(3) of lemma~\ref{lma:rot} $h$ vanishes on some open interval, and so by quasianalyticity it vanishes everywhere.
%Thus the Fourier coefficient $S_{\abs{k}}(\gamma)\A_{\abs{k}} a_k(z_\gamma)$ vanishes for every $k\in\Z$.
Thus each Fourier coefficient $S_{\abs{k}}(\gamma)\A_{\abs{k}} a_k(z_\gamma)$ is zero.
}

If $\iota_\gamma=-\kappa_\gamma\notin\pi\Q$, we have $S_k(\gamma)\neq0$ for all $k$ by lemma~\ref{lma:S_k} and so $\A_k a_k(z_\gamma)=0$ for all $k$.
The set of such $z_\gamma$ that $\iota_\gamma=-\kappa_\gamma\notin\pi\Q$ and $\gamma\in\Gamma_E$ is dense on $[0,1]$, so continuity of $\A_ka_k$ (lemma~\ref{lma:Abel}) implies that $\A_k a_k=0$. This, by the injectivity of $\A_k$ (lemma~\ref{lma:Abel}), implies that each $a_k$ is identically zero.

We have thus demonstrated that the function $f$ is antisymmetric with respect to the normal to $\partial D$ at angle zero. The same conclusion may be drawn with any line intersecting $E$ and the origin. If $f$ is antisymmetric with respect to two different lines going through the origin, $f$ vanishes everywhere. %This completes the proof.
\end{proof}

Notice that by proposition~\ref{prop:spt} (in dimension two) $f\in C$ with $\brt f=0$ vanishes in $\bar{D}\setminus\co(\partial D\setminus E)$. If $f$ is also  quasianalytic in the angular variable, it follows that $f$ vanishes outside the open disk of radius $d(0,\bar{D}\setminus\co(\partial D\setminus E))$ centered at the origin. The fact that $f$ also vanishes inside this disk, as shown above, is nontrivial. We also remark that we have not employed the support theorem in our proof of theorem~\ref{thm:open}.

\begin{proof}[Proof of corollary~\ref{cor:open}]
Take any two dimensional subspace $F$ of $\R^n$ such that it meets $E$. As noted in the proof of corollary~\ref{cor:one} at the end of section~\ref{sec:one-pf}, one ends up with the two dimensional broken ray transform in the disk $\bar{B}^n\cap F$. Since $E\cap F$ is open in the circle $S^{n-1}\cap F$, theorem~\ref{thm:open} guarantees that $\brt f$ restricted to those $\gamma$ that lie in $F$ uniquely determine $f$ in $\bar{B}^n\cap F$. Since $\R^n$ can be written as a union of such subspaces $F$, $\brt f$ determines $f$ everywhere.
\end{proof}

\section*{Acknowledgements}
The author is partly supported by the Academy of Finland (no 250~215). The author wishes to thank Mikko Salo for many discussions and insightful comments regarding this work and the referees for useful feedback.

\bibliographystyle{amsplain}
\bibliography{disk}

\end{document}